\documentclass[a4paper, 11pt]{article}

\usepackage{amsmath} 
  \usepackage{amsthm}
\usepackage{enumerate}
\usepackage{optidef}
\usepackage{dsfont}
\usepackage{hyperref}
\usepackage{cite} 

\title{\bf
Preserving Privacy of the Influence Structure in Friedkin-Johnsen Systems
}

\author{Jack Liell-Cock, Ian R. Manchester$^*$ and Guodong Shi
\thanks{I. R. Manchester and G. Shi are with Australian Centre for Field Robotics, The University of Sydney, NSW 2006, Australia. (email: ian.manchester@sydney.edu.au, guodong.shi@sydney.edu.au) }%
}

\date{}

\newtheorem{theorem}{Theorem}
\newtheorem{proposition}{Proposition}
\newtheorem{lemma}{Lemma}
\theoremstyle{definition}
\newtheorem{definition}{Definition}
\newtheorem{example}{Example}
\newtheorem{claim}{Claim}
\newcommand{\diag}{\mathrm{diag}}
\newcommand{\N}{\mathrm{N}}

\def\baselinestretch{1.3}
\begin{document}

\maketitle

\begin{abstract}

The nature of information sharing in common distributed consensus algorithms permits network eavesdroppers to expose sensitive system information. An important parameter within distributed systems, often neglected under the scope of privacy preservation, is the influence structure -- the weighting each agent places on the sources of their opinion pool. This paper proposes a local (i.e. computed individually by each agent), time varying mask to prevent the discovery of the influence structure by an external observer with access to the entire information flow, network knowledge and mask formulation. This result is produced through the auxiliary demonstration of the preserved stability of a Friedkin-Johnsen system under a set of generalised conditions. The mask is developed under these constraints and involves perturbing the influence structure by decaying pseudonoise. This paper provides the information matrix of the best influence structure estimate by an eavesdropper lacking a priori knowledge and uses stochastic simulations to analyse the performance of the mask against ranging system hyperparameters.

\end{abstract}

\section{Introduction}

In distributed systems, many independent, sparsely connected agents work together to achieve a common goal. Due to advancements in communication and automation technology, many systems have begun to adopt a decentralised approach such as intelligent transportation networks \cite{DeWit2015}, smart grids \cite{Hill2012} and the Internet of Things \cite{Chang2015}. In large distributed networks, local controllers do not guarantee the optimisation or stability of the system, while centralised controllers cannot be suitably scaled or are physically impractical to implement. Thus emerged the branch of distributed control, the origin of which can be traced from the work of distributed decision-making \cite{Tsitsiklis1986} and parallel computation \cite{Lynch1996}.

A simple tool developed in distributed control theory are consensus algorithms \cite{DeGroot1973, Hegselmann2002, Boyd2006, Amelkin2017} -- local algorithms which require minimal computation, communication and synchronisation to aggregate the parameters within a distributed system. In fact, common optimisation and estimation problems such as least squares, sensor calibration, vehicle coordination and Kalman filtering can be formulated as averages of some system parameters \cite{Garin2010}, and thus can be solved using consensus algorithms.

Many consensus algorithms stem from the DeGroot model \cite{DeGroot1973}, which conveys how a set of agents may reach a consensus by opinion broadcasting. In the context of distributed systems, an opinion is an agent's evaluation of a parameter within the system. At each iteration, the opinions are broadcast, and each agent updates their opinion to a weighted aggregate of their local neighbours' ones. The weightings that each agent places on the opinions of other agents is called the influence structure.

An extension of the DeGroot model which employs heterogeneity is the Friedkin-Johnsen (FJ) model \cite{FJ2011}. This scheme was originally developed to model the propagation of opinions through social networks, because positively interacting networks can still persistently disagree and cluster if their agents are diverse \cite{Xia2011, Aeyels2011}. This model generalises the DeGroot model with each agent retaining a constant prejudice towards an external bias -- commonly chosen as the agent's initial opinion. The weighting the agents place on their biased opinions versus the opinions of other agents is called their susceptibility. A low susceptibility indicates the agent is stubborn in their original opinion whereas a high susceptibility means the agent is more gullible to the opinions of others. We define a distributed system which uses the FJ model to reach a consensus as an FJ system.

While consensus algorithms enjoy system scalability, these approaches are vulnerable to network eavesdroppers due to their information sharing nature. Considerable results have been established to preserve the privacy of the initial opinions of the agents \cite{Duan2015, Mo2017, Liu2018, Altafini} by injecting noise or randomly delaying the broadcasts of the true opinions. However, literature still fails to recognise the complications arising from the exposure of the influence structure by an eavesdropping agent. For example, in smart grid technology, a range of electricity producers and consumers are interconnected to form an efficient energy exchange system \cite{Hill2012}. It is likely that an electricity producer would use an influence structure to aggregate their consumer's needs and generate power accordingly. The exposure of this influence structure would unveil an unfair bargaining advantage for the consumers by indicating how much the producer relies on their energy consumption. Similar privacy concerns are prevalent in the exchange of information in the smart car industry, and the home automation industry with the Internet of Things.

The aim of this paper is to develop a mask for the influence structure of an FJ system. The term \textit{mask} was coined in \cite{Altafini} and describes a local (in the sense it can be computed on an agent by agent basis) function which induces noise  into some system parameter. The purpose of a mask is to preserve the privacy of the parameter, without affecting the limiting conditions of the system. Our mask development process originates by establishing a range of perturbations applicable to the parameters of an FJ system which don't affect the resulting aggregated opinions. The mask is then formed within these restrictions and consists of randomly offsetting the true influence structure by noise from a decaying normal distribution. The extent to which a network eavesdropper with access to the entire information flow, network knowledge and mask formulation can discover the influence structure is then carefully explored. The key results from this investigation are the information matrix produced by the eavesdropper's maximum likelihood estimate, and numerical simulations outlining the dependencies of the mask's performance with respect to FJ system hyperparameters.

The remainder of this paper is organised as follows. In Section \ref{sec:preliminaries}, we outline notation and revisit the problem formulation. In Section \ref{sec:fjstability}, we extend the stability conditions of the FJ system to allow for non-constant influence structures, susceptibilities and external biases. Section \ref{sec:mask} uses this result to develop a mask for the influence structure in an FJ system. Section \ref{sec:discoverability} precisely defines the system eavesdropper and performs maximum likelihood estimation on the influence structure from its perspective. Numerical results are also produced to aid the discussions on the mask's performance. Finally, Section \ref{sec:conclusion} ends the paper with some concluding remarks and potential future directions.

\section{Preliminaries and Problem Definition} \label{sec:preliminaries}

\subsection{Notation}

Discrete, time-varying sequences of vectors and matrices are a common theme throughout this paper. To maintain consistency across index notation, a superscript index is an index with respect to the elements (or agents) in a vector or matrix, while a subscript index is an index with respect to time. For example, given a discrete-time vector sequence of evolving opinions $(x_t)_{t\in\mathds{N}}$, the value $x^i_t$ is the opinion of the $i$-th agent at timestep $t$.

The notation $\diag(A_1,...,A_m)$ denotes the block diagonal matrix constructed using matrices $\{A_k \ | \ k=1,...,m\}$ with zeros elsewhere. A non-negative matrix that has all rows sum to 1 is called a row stochastic matrix. For any square matrix $A$, we say $A \succ 0$ if and only if $A$ is positive definite. All other vector or matrix inequalities, for example $A \geq B$, are taken elementwise for $B$ the same size as $A$, or $B$ a scalar. For matrix $Q$, $Q \in N(A)$ if its columns form an orthonormal basis for the nullspace of $A$.

Let $\|v\|_p = \sqrt[p]{\sum_{k} |v^k|^p}$ be the $l_p$-norm of the vector $v$ with $\|v\|=\|v\|_2$. We define $\mathds{1}_m$ as the vector of length $m$ with 1 as all its entries, and $I_{m}$ as the $m \times m$ identity matrix. If $m$ is not specified, the size of the vector or matrix is assumed from the context of the equation. Let $\mathcal{N}(\mu, \sigma^2)$ be the normal distribution with mean $\mu$ and variance $\sigma^2$.

Consider a set of agents $\mathcal{V}$ with communication links $\mathcal{E} \subseteq \mathcal{V} \times \mathcal{V}$ whose interactions are described by a graph $\mathcal{G} = (\mathcal{V},\mathcal{E})$. This is referred to as a network. Let the total number of agents in the network be $|\mathcal{V}| = n < \infty$. An edge is drawn $(i,j) \in \mathcal{E}$ if agent $j$ directly influences the opinion of agent $i$. A path from $i$ to $j$ in $\mathcal{G}$ is an ordered set of edges $((i,k_1), (k_1,k_2),...,(k_l,j))$. If such a path exists, we say $j$ indirectly influences $i$. The neighbourhood $\mathcal{N}^i = \{j \ | \ (i,j) \in \mathcal{E} \}$ is the set of agents which directly influence agent $i$, and the degree $d^i = |\mathcal{N}^i|$ is the size of this neighbourhood. If there exists a $\tilde{d}$ such that $d^i = \tilde{d}$ for all $i \in \mathcal{V}$, then the network or associated system is said to have degree $\tilde{d}$. Finally, $W$ is said to be a matrix adapted to the graph $\mathcal{G}$ if and only if $(i,j) \not\in \mathcal{E}$ implies $W^{ij} = 0$. The remaining elements $W^{ij}$ where $(i,j) \in \mathcal{E}$ are denoted influence elements.

\subsection{The FJ Model}

Consider a network of agents $\mathcal{G} = (\mathcal{V},\mathcal{E})$. Let each agent hold a state $x_t \in \mathds{R}^n$ which represents the opinions of the agents at timestep $t \in \mathds{N}$. The influence structure $W$ is the row stochastic matrix adapted to the network $\mathcal{G}$ which holds the relative weightings each agent places on the other agents. We call the $i$-th row of $W$ the influence structure of agent $i$. Define the diagonal matrix $0 \leq \Lambda \leq I_n$ holding the ordered susceptibilities $\lambda\in\mathds{R}^n$ of each agent as the susceptibility matrix, and its non-negative complement $\bar\Lambda = I_n-\Lambda$ as the stubbornness matrix. An agent $i\in\mathcal{V}$ is said to be oblivious if $\lambda^i = 1$.  Let $u\in \mathds{R}^n$ be the external biases for each agent. Then the opinions in an FJ system evolve as
\begin{equation}
x_{t+1} = \Lambda W x_t + \bar\Lambda u. \label{eq:FJ}
\end{equation}

The stability conditions of FJ systems are given in the following result from \cite{Parsegov2017}.

\begin{proposition} \label{prop:fj_stability}
	For every agent $i \in \mathcal{V}$, if some agent $j\in\mathcal{V}$ is not oblivious and indirectly influences $i$, the FJ system is stable with the opinions converging to
	\begin{equation}
	x_\infty = \lim\limits_{t\to\infty} x_t = (I-\Lambda W)^{-1}\bar\Lambda u.
	\end{equation}
	Moreover, this is equivalent to $\Lambda W$ being Schur stable. 
\end{proposition}

\subsection{Problem Formulation} \label{sec:problem_formulation}

The ability for a resourced eavesdropper to expose the influence structure of an FJ system is outlined in the following results.

\begin{lemma} \label{lem:lambda_W_identifiability}
    If $\Lambda W$ is identifiable, then $\lambda$ and $W$ are identifiable.
\end{lemma}

\begin{proof}
    Let $i\in\mathcal{V}$ be an arbitrary agent. From the assumptions, the values of $\lambda^iW^i$ are known. By the row stochastic nature of $W$, $\lambda^i = \|\lambda^iW^i\|_1$. The influence structure of agent $i$ can then be identified by $W^i = \frac{1}{\lambda^i} \cdot \lambda^i W^i$. Repeating this for all agents completes the proof.
\end{proof}

\begin{proposition} \label{prop:unmasked_discovery}
	For a sufficiently excited FJ system, an eavesdropper with access to the opinion trajectories of all agents after time $T > 0$ can identify the influence structure.
\end{proposition}

In this case, sufficiently excited means that no opinion trajectory can be reconstructed as a linear combination of the others. This is satisfied when the biases and initial opinions of the agents have an independent, random aspect to them. Since the agents are considered independent entities, we assume this is true.

\begin{proof}
	Using the evolution of the FJ system in \eqref{eq:FJ}, the dependence on the external bias can be removed by taking the difference between two successive updates,
	\begin{equation*}
	\tilde{x}_{t} = \Lambda W \tilde{x}_{t-1},
	\end{equation*}
	where $\tilde{x}_{t} = x_{t+1} - x_{t}$. Since the agents are sufficiently excited, $\Lambda W$ can be identified using linear regression after at least $n+2$ successive measurements of $x_t$,
	\begin{equation}
	\Lambda W \tilde{X}_{T:T+n-1} = \tilde{X}_{T+1:T+n}, \label{eq:fjunmasked}
	\end{equation}
	where $\tilde{X}_{s:t} = \left[\tilde{x}_{s}, ..., \tilde{x}_{t}\right]$. Then from Lemma \ref{lem:lambda_W_identifiability}, $W$ and $\lambda$ can be individually identified.
\end{proof}

The following example instantiates the above results by identifying the influence structure and susceptibilities of a simple 3 agent FJ system using observations of the opinion trajectories for $1 \leq t\leq 5$.

\begin{example}
	Define an FJ system by
	\begin{equation*}
	    W = \begin{bmatrix}
	    0 & 0.5 & 0.5 \\ 0.2 & 0.2 & 0.6 \\ 0.5 & 0 & 0.5
	    \end{bmatrix}, \enskip \lambda = \begin{bmatrix}
	    0.4 \\ 0.5 \\ 0.6
	    \end{bmatrix}, \enskip x_0 = u = \begin{bmatrix}
	    1 \\ 2 \\ 3
	    \end{bmatrix}.
	\end{equation*}
	The opinion trajectories for $1\leq t \leq 5$ are
	\begin{equation*}
	    \begin{bmatrix}
	    1.6 & 1.52 & 1.5 & 1.4916 & 1.48676 \\
	    2.2 & 2.1 & 2.082 & 2.071 & 2.0651 \\
	    2.4 & 2.4 & 2.376 & 2.3628 & 2.35632
	    \end{bmatrix},
	\end{equation*}
	which produces
	\begin{align*}
	    \tilde{X}_{1:3} &= \begin{bmatrix}
        -0.08 & -0.02 & -0.0084 \\
        -0.1 & -0.018 & -0.011 \\
        0 & -0.024 & -0.0132
	    \end{bmatrix}, \\ \tilde{X}_{2:4} &= \begin{bmatrix}
	    -0.02 & -0.0084 & -0.00484 \\
        -0.018 & -0.011 & -0.0059 \\
        -0.024 & -0.0132 & -0.00648
        \end{bmatrix}.
	\end{align*}
	We can identify $\Lambda W$ by
	\begin{equation*}
	    \Lambda W = \tilde{X}_{2:4}\left(\tilde{X}_{1:3}\right)^{-1} = \begin{bmatrix}
	    0 & 0.2 & 0.2 \\ 0.1 & 0.1 & 0.3 \\ 0.3 & 0 & 0.3
	    \end{bmatrix}.
	\end{equation*}
	Finally using Lemma \ref{lem:lambda_W_identifiability}, we can individually calculate $\lambda$ and $W$ on a row by row basis,
	\begin{equation*}
	    \lambda = \|W\|_1 = \begin{bmatrix}
	    0.4 \\ 0.5 \\ 0.6
	    \end{bmatrix}, \enskip W = \Lambda W / \lambda = \begin{bmatrix}
	    0 & 0.5 & 0.5 \\ 0.2 & 0.2 & 0.6 \\ 0.5 & 0 & 0.5
	    \end{bmatrix}.
	\end{equation*}
\end{example}

We define a new metric to quantify the performance of an influence structure mask.

\begin{definition}
	The estimate error is the square root of the largest eigenvalue of the covariance for a maximum likelihood estimate of the influence structure. This is equivalent to the inverted square root of the E-optimal experiment design criterion.
\end{definition}

Similar privacy metrics using maximum likelihood estimation are common for analysing differential privacy schemes such as in \cite{Mo2017}. This metric is a valid measure of privacy because it is the expected error between the elements from a best estimate of $W$ given no a priori knowledge. Since the range of permissible values of $W$ is between 0 and 1, if the expectation of the estimate error is greater than 1, the influence structure is said to be undiscoverable. Otherwise, we call the influence structure discoverable. The goal of this paper is to develop a mask to make the influence structure of an FJ system undiscoverable.

\section{A Family of Masked FJ Models} \label{sec:fjstability}


Popular schemes for privacy preservation perturb the system behaviour in various ways, so there is an inherent tension between secrecy and maintaining normal function. Therefore, it is useful to characterise the stability of FJ systems under a wide range of distortions. This establishes a ``development environment'' outlining a set of rules a mask must follow to preserve the limiting conditions of the system. Theorem \ref{thm:fjtheorem} provides this classification for a range of disturbances to the parameters of an FJ system. 

\begin{theorem}
	\label{thm:fjtheorem}
	Given an FJ system \eqref{eq:FJ} which is stable with limit $x_\infty$, and an arbitrary set of element-wise matrix and vector sequences
	\begin{equation*}
	\left. \begin{matrix}
	W_t \to W \\
	\Lambda_t \to \Lambda \\
	u_t \to u
	\end{matrix} \right\} \text{as} \ t \to \infty,
	\end{equation*}
	with $u_t$ converging exponentially, the system
	\begin{equation}
	z_{t+1} = \Lambda_t W_t z_t + \bar\Lambda_tu_t \label{eq:FJk}
	\end{equation}
	converges to $x_\infty$ independent of the initial opinions.
\end{theorem}

\begin{proof}
	From Proposition \ref{prop:fj_stability}, an FJ system is stable if and only if the matrix $\Lambda W$ is Schur stable. Given that $\Lambda W$ is Schur stable, for some symmetric positive definite matrix $R\in\mathds{R}^{n\times n}$, there exists a symmetric positive definite matrix $P \in \mathds{R}^{n\times n}$ such that \cite{Hespanha2018}
	\begin{equation*}
	P - (\Lambda W)'P(\Lambda W) = R.
	\end{equation*}
	As $P$ and $R$ are both positive definite, there exists some $\tilde\alpha \in (0,1)$ such that
	\begin{equation*}
	R - \tilde\alpha P \succ 0,
	\end{equation*}
	which implies
	\begin{equation}
	(1-\tilde\alpha)P - (\Lambda W)'P(\Lambda W) \succ 0. \label{eq:in_m}
	\end{equation}
	Define the set of matrices
	\begin{equation*}
	\mathcal{M} := \{M \in \mathds{R}^{n\times n} \ | \ (1-\tilde\alpha)P - M'PM \succ 0\}.
	\end{equation*}
	Since the set of positive definite matrices is open, and the function $(1-\tilde\alpha)P - M'PM$ is continuous with respect to $M$, $\mathcal{M}$ is an open set. Also, from \eqref{eq:in_m}, $\Lambda W \in \mathcal{M}$. By the multiplication rule for limits $\Lambda_t W_t \to \Lambda W \in \mathcal{M}$, hence there exists a $t_0 \in \mathds{N}$ such that for all $t \geq t_0$, $\Lambda_t W_t \in \mathcal{M}$.
	
	Define a discrete time system as
	\begin{equation}
	x_{t+1} = \Lambda_t W_t x_t := f_t(x_t). \label{eq:lyapunov_system}
	\end{equation}
	As $P$ is positive definite, the function $V: \mathds{R}^n \to \mathds{R}$ defined by $V(x)=\sqrt{x'Px}$ is a vector norm. $V$ is a continuous, positive definite, radially unbounded function, and for all $t \geq t_0$ and all $x \in \mathds{R}^n \setminus \{0\}$,
	\begin{align*}
	V(f_t(x)) &= \sqrt{x'(\Lambda_t W_t)'P(\Lambda_t W_t) x} \\
	&< \sqrt{x (1-\tilde\alpha)Px} \\
	&= \alpha V(x),
	\end{align*}
	where $\alpha := \sqrt{1-\tilde\alpha} < 1$. Hence, $V$ is a well-defined Lyapunov function for \eqref{eq:lyapunov_system} when $t \geq t_0$.
	
	For $z_t$ defined by \eqref{eq:FJk}, $z_{t_0}$ will remain finite provided $t_0$ is finite. Moreover, the convergence properties of a sequence are independent of the first finite terms. Therefore, for the remainder of the proof, the analysis of system \eqref{eq:FJk} is reduced to the analysis on $(z_t)_{t \geq t_0}$. Using the triangle inequality,
	\begin{equation}
	V(z_{t+1}) = V(f_t(z_t) + \bar\Lambda_t u_t) \leq V(f_t(z_t)) + V(\bar\Lambda_t u_t). \label{eq:bounded_proof1}
	\end{equation}
	By the multiplication rule for limits $\bar\Lambda_t u_t \to \bar\Lambda u$, hence there exists an $m_1 > 0$ such that for all $t \in \mathds{N}$, $V(\bar\Lambda_t u_t) < m_1$. Continuing from \eqref{eq:bounded_proof1} gives
	\begin{align}
	V(z_{t+1})
	&<\alpha V(z_t) + m_1.
	\end{align}
	Therefore, for all $V(z_t) > \frac{m_1}{1-\alpha}$,
	\begin{equation}
	V(z_{t+1}) < \alpha V(z_t) + (1-\alpha) V(z_t) = V(z_t).
	\end{equation}
	So $V(z_t)$ is strictly decreasing when it is larger than $\frac{m_1}{1-\alpha}$, implying that $z_t$ is bounded. 

	Define the sequence $y_t = z_{t+1} - z_t$. By the triangle inequality,
	\begin{align*}
	V(y_t)
	&= V(\Lambda_t W_t z_t + \bar\Lambda_t u_t - \Lambda_{t-1}W_{t-1}z_{t-1} - \bar\Lambda_{t-1}u_{t-1}) \\
	&= V(\Lambda_t W_t z_t - \Lambda_t W_t z_{t-1} \\
	&\qquad + \Lambda_t W_t z_{t-1} - \Lambda_{t-1}W_{t-1}z_{t-1} \\
	&\qquad + \bar\Lambda_t u_t - \bar\Lambda_{t-1}u_{t-1}) \\
	&\leq V(f_t (y_{t-1})) + V([\Lambda_t W_t - \Lambda_{t-1}W_{t-1}]z_{t-1}) \\
	&\qquad + V(\bar\Lambda_t u_t - \bar\Lambda_{t-1}u_{t-1}).
	\end{align*}
	As $z_t$ is bounded and $\Lambda_tW_t$ contracts exponentially under $V$, $V([\Lambda_t W_t - \Lambda_{t-1}W_{t-1}]z_{t-1}) \to 0$ exponentially. Additionally, $V(\bar\Lambda_t u_t - \bar\Lambda_{t-1}u_{t-1}) \to 0$ exponentially from the exponential convergence of $u_t$. Therefore, there exists some $\beta < 1$ and $m_2 > 0$ such that,
	\begin{align*}
	m_2\beta^t &> V([\Lambda_t W_t - \Lambda_{t-1}W_{t-1}]z_{t-1})\\
	&\qquad +V(\bar\Lambda_t u_t - \bar\Lambda_{t-1}u_{t-1}).
	\end{align*}
	Combining these results gives
	\begin{equation*}
	V(y_t) < \alpha V(y_{t-1}) + m_2\beta^t.
	\end{equation*}
	When $V(y_{t-1}) > \frac{2m_2\beta^t}{1-\alpha}$,
	\begin{equation*}
	V(y_t) < \alpha V(y_{t-1}) + \frac{1-\alpha}{2} V(y_{t-1}) = \frac{1+\alpha}{2} V(y_{t-1}).
	\end{equation*}
	So $V(y_t)$ falls below $m_2\beta^t$ at an exponential rate. Since $m_2\beta^t \to 0$ exponentially, $V(y_t) = V(z_{t+1}-z_t) \to 0$ exponentially.  Thus $z_t$ converges to some $z_\infty \in \mathds{R}^{n}$. Taking the limit of both sides of \eqref{eq:FJk} gives
	\begin{equation*}
	z_\infty = \Lambda W z_\infty + \bar{\Lambda} u
	\implies z_\infty = (I - W\Lambda)^{-1}\bar\Lambda u = x_\infty,
	\end{equation*}
	completing the proof.
\end{proof}

\section{Influence Structure Mask} \label{sec:mask}


In this section, we present a specific realisation of the class defined in Theorem \ref{thm:fjtheorem} to mask the values of the influence structure in an FJ system. As indicated by the theorem, this mask offsets the true influence structure to produce confidentiality. The susceptibilities and the external biases remain unchanged because we assume the worst-case scenario that the eavesdropper is aware of these values. Additionally, the non-influence elements of $W$ are unaltered because the existence of the corresponding connections in the underlying graph is not certain.

By motivation from differential privacy ideas explored in \cite{Mo2017,Altafini}, the mask we developed is
\begin{equation}
W_t = W + e^{-\varphi t} V_t, \label{eq:Wmask}
\end{equation}
where $\varphi > 0$ is called the decay rate and is publicly known. $V_t$ is a matrix adapted to the network $\mathcal{G}$ with the remaining influence elements independently chosen from $\mathcal{N}(0,1)$. Each row of $V_t$ is chosen locally and privately by the corresponding agent, allowing the algorithm to be compatible with distributed systems. Compared to \cite{Mo2017,Altafini}, we propose to add a mask mechanism to the weights instead of the dynamic states of the system. While conceptually similar, the significant difference is that our mask leads to a time-varying system, and the others a time-invariant system with time-varying noises. 

At each time step, the following actions are taken:
\begin{enumerate}[(i)]
	\item Each agent populates their elements of $V_t$ with independently selected random variables chosen from $\mathcal{N}(0,1)$.
	\item The decoy influence structure is set as $W_t = W + e^{-\varphi t}V_t$.
	\item The opinions are pooled: $x_{t+1} = \Lambda W_tx_{t} + \bar\Lambda u$.
\end{enumerate}
$W_t$ converges to $W$ by the decay term, hence Theorem \ref{thm:fjtheorem} enforces that the system
\begin{equation*}
x_{t+1} = \Lambda W_t x_t + \bar{\Lambda} u
\end{equation*}
converges to the same consensus as \eqref{eq:FJ}, provided \eqref{eq:FJ} is stable.


The idea behind the mask is that the decay rate $\varphi$, is slower than the convergence of the agent's opinions to their consensus value, so the opinions convergence is dominated by the mask fading rather than the dynamics of the FJ system. Figure \ref{fig:iterative_offset_mask} illustrates the application of this mask to a standard 5 agent FJ system.

\begin{figure}[htb!]
	\centering
	\includegraphics[width=.6\columnwidth]{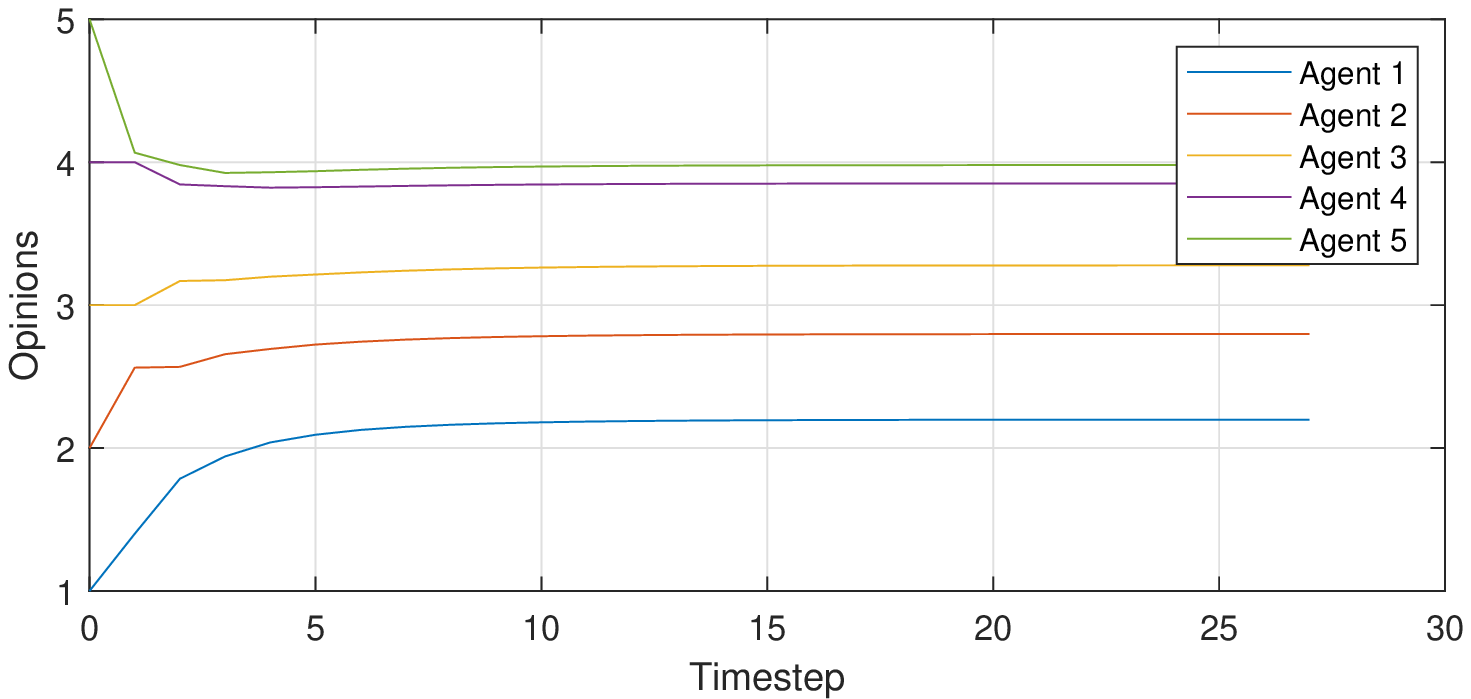}
	\includegraphics[width=.6\columnwidth]{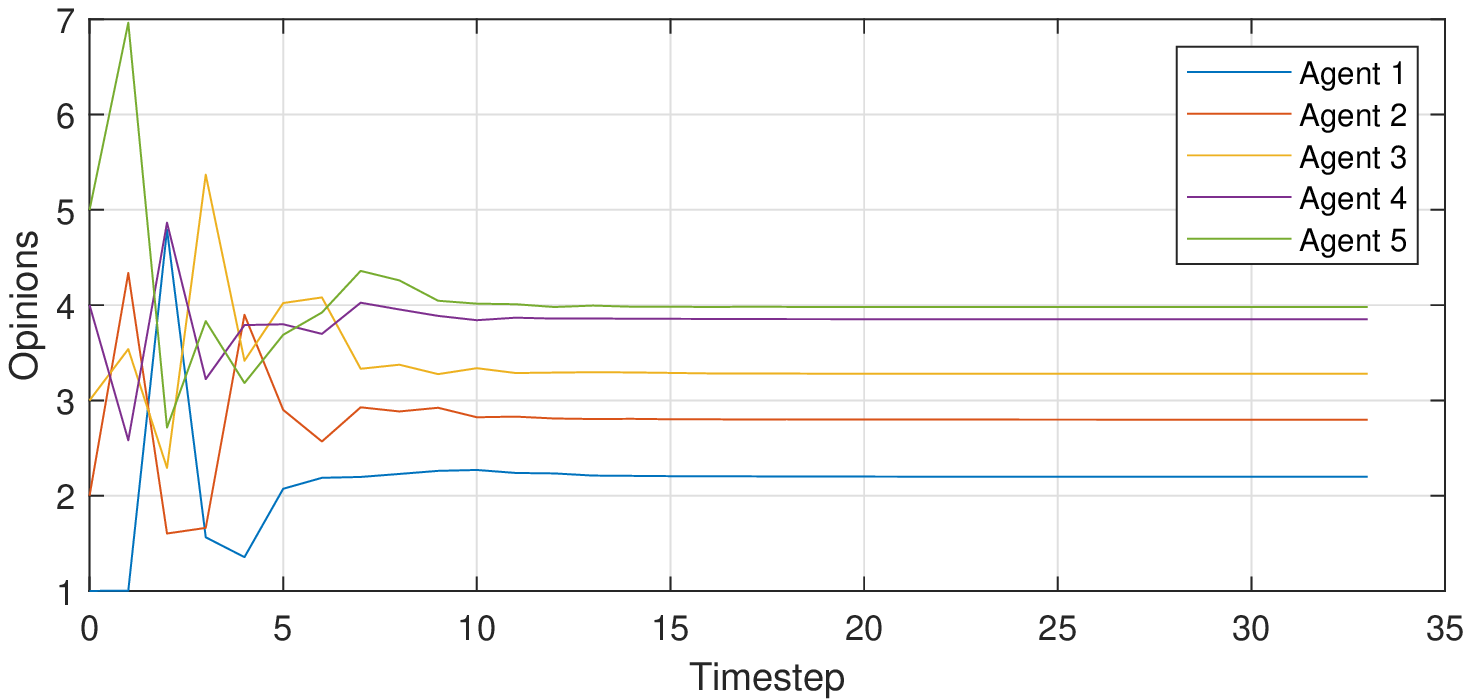}
	\caption{Comparison of FJ system without (top) and with (bottom) the implementation of the mask}
	\label{fig:iterative_offset_mask}
\end{figure}

\section{System Discoverability} \label{sec:discoverability}

We aim to prevent the influence structure of any agent from being directly measured or inferred by observers who gain access to the data flow and structure of the distributed system during the opinion pooling. In this section, we introduce the eavesdropper of concern which enjoys a worst-case knowledge set of the distributed system parameters. We determine the information matrix for the maximum likelihood estimate by this eavesdropper and use numerical results to discuss threats to the security of the distributed system. The maximum likelihood estimation problem is an extension of \cite{Ravazzi2018} where the influence structure is estimated from a range of initial biases and limiting opinions resulting from multiple opinion poolings over the same FJ network.

\subsection{The Eavesdropper} \label{sec:eavesdropper}

The eavesdropper attempts to estimate the parameter $W$ through observations of $x_t$ and the remaining parameters of the FJ system. We assume the formula of the mask and the distribution of the randomly generated elements of $V_t$ are public knowledge. The knowledge space of the eavesdropper is therefore
\begin{equation*}
\mathcal{K} = \left\{\mathcal{G}, u, \Lambda, \{x_0, ..., x_T\} \right\},
\end{equation*}
where $T$ is the timestep when a consensus is reached, and the agents cease their opinion broadcasts.

With the mask, the trajectory of opinions in the FJ system is given by \eqref{eq:Wmask}. Rearranging this produces
\begin{equation}
W_tx_t = \Lambda^{-1} \left( x_{t+1} - \bar\Lambda u\right). \label{eq:W_linear_expression}
\end{equation}
Without loss of generality, we assume the eavesdropper is seeking the influence structure of the first agent. Specifically, the eavesdropper must compute the values in the first row of $W$. $W$ is adapted from the underlying network $\mathcal{G}$, so any elements of $W$ representing disconnections are known by the eavesdropper to be zero. Therefore, exposing the influence structure of the first agent reduces to the problem of calculating the influence elements in the first row of $W$. 

To accommodate notation, let $w$, $w_t$ and $v_t$ be vectors holding the influence elements of the first rows of $W$, $W_t$ and $V_t$ respectively. Consequently,
\begin{equation}
w_t = w + e^{-\varphi t}v_t. \label{eq:wmask}
\end{equation}

Let $A_t$ be the row vector containing the opinions $x_t$ which directly influence the opinion of the first agent, and $B_t = \Lambda^{-1} \left( x_{t+1} - \bar\Lambda u\right)$ be the right-hand side of \eqref{eq:W_linear_expression}. Taking $b_t$ as the first element of $B_t$,
\begin{equation}
A_t w_t = b_t. \label{eq:linalg}
\end{equation}
The number of influence elements in the first row of $W$ is equivalent to the degree of the first agent $d^1$, which we denote $d$ for ease. It follows that the vectors $w$, $w_t$, $v_t$ and $A_t$ all have $d$ elements.

\subsection{Maximum Likelihood Estimation}

The eavesdropper understands that $w_t$ is generated by \eqref{eq:wmask}, where the elements of $v_t$ are selected from $\mathcal{N}(0,1)$. Therefore, the elements of $e^{-\varphi t} v_t$ can be considered to be selected from the distribution $\mathcal{N}(0,e^{-2\varphi t})$. The probability density function of this vector is
$$p(e^{-\varphi t} v_t) = \left(\frac{1}{e^{-\varphi t}\sqrt{2\pi}} \right)^{d} \exp \left(-\frac{\| e^{-\varphi t}v_t \|^2}{2e^{-2\varphi t}} \right).$$
This can be treated as a likelihood function of $w$ given an observation $w_t$ using \eqref{eq:wmask}.
$$
L(w \ | \ w_t) = \left(\frac{1}{e^{-\varphi t}\sqrt{2\pi}} \right)^{d} \exp \left( - \frac{ \| w_t - w \|^2 }{2e^{-2\varphi t} } \right)
$$
Ignoring the scalar term, the log-likelihood of $w$ given $w_t$ is
$$
l(w \ | \ w_t) =  -\frac{ \| w_t - w \|^2 }{2e^{-2\varphi t} }.
$$
Therefore, the log-likelihood function of $w$ for a range of observations is
$$
l(w \ | \ w_0,...,w_{T-1}) =  -\sum_{t=0}^{T-1}\frac{ \| w_t - w \|^2 }{2e^{-2\varphi t} }.
$$

The values of $w_t$ cannot be directly measured by the eavesdropper. Instead, each observation provides a restriction for $w_t$ expressed by the constraint \eqref{eq:linalg}. Additional constraints on $w$ are $\mathds{1}'w = 1$ and $w \geq 0$ because $W$ is row stochastic.

Using all information available to the eavesdropper, the maximum likelihood estimate for $w$ can be formulated as the solution to a quadratic optimisation problem.
\begin{maxi}|l|
	{w, w_t}{-\sum_{t=0}^{T-1} \frac{ \| w_t - w \|^2 }{2e^{-2\varphi t} }}{}{}
	\addConstraint{A_tw_t}{= b_t}{}
	\addConstraint{\mathds{1}'w}{= 1}
	\addConstraint{w}{\geq 0}{} \label{eq:first_optimisation}
\end{maxi}

Define the vector $\tilde{w} = \left[w_0', ..., w_{T-1}', w'\right]'$ to hold the variables of \eqref{eq:first_optimisation}, and define $
A = \diag(A_0, ..., A_{T-1}, \mathds{1}')$. Similarly, define the vector $ b = \left[b_0,...,  b_{T-1}, 1 \right]'$. Therefore, the set of linear equality constraints can be expressed as $A\tilde{w} = b$.

Define the diagonal matrix $H_k = \diag(H_{0,k},...,H_{T-1,k})$, where $H_{t,k} = e^{2\varphi t} I_k$ are the identity matrices iteratively multiplied by the scaling factors of the cost function in \eqref{eq:first_optimisation}. Also define
\begin{equation*}
Y = \begin{bmatrix}
I_d & & & -I_d \\
& \ddots & & \vdots \\
& & I_d & -I_d
\end{bmatrix} \in \mathds{R}^{Td \times (T+1)d}.
\end{equation*}
It follows that $l(\tilde{w}) = -1/2 \cdot \tilde{w}'Y'H_dY\tilde{w}$, so the maximum likelihood problem can be reduced to:
\begin{maxi*}|l|
	{\tilde{w}}{-\frac{1}{2} \cdot \tilde{w}'Y'H_dY\tilde{w}}{}{}
	\addConstraint{A\tilde{w}}{= b}
	\addConstraint{w}{\geq 0.}
\end{maxi*}

Define $Q = \diag(Q_0,...,Q_{T-1}, Q_\mathds{1})$ where $Q_t \in \N(A_t)$ and $Q_\mathds{1} \in \N(\mathds{1}')$. Therefore, $Q \in \N(A)$. The information matrix of the maximum likelihood estimate is the negative curvature of the cost function restricted to the constraints, which is given by
\begin{equation*}
\mathcal{I} = \frac{\partial^2}{\partial s^2} \left(\frac{1}{2}\cdot sQ'Y'H_dYQs\right) = Q'Y'H_dYQ.
\end{equation*}

Although $\mathcal{I}$ gives the Fisher information for the estimate quality of $\tilde w$, only the information for $w$ is of importance because the variables $w_t$ were merely constructed to aid in its recovery. It is of no concern to the eavesdropper if the information about the variables $w_t$ is minimal, provided the estimate for $w$ is accurate.

By block matrix inversion \cite{Lu2002}, the information matrix of $w$ is the Schur compliment of the bottom right $(d-1)\times (d-1)$ block of $\mathcal{I}$. Expanding and simplifying the expression for $\mathcal{I}$ yields
\begin{equation*}
\mathcal{I} = \begin{bmatrix} 
H_{d-1} & K \\ K' & R
\end{bmatrix}
\end{equation*}
where $R = \sum_{t=0}^{T-1}e^{2t\varphi} I_{d-1}$, and
\begin{equation*}
K = \begin{bmatrix}
-Q_0'Q_{\mathds{1}} \\ -e^{2\varphi} Q_1'Q_{\mathds{1}} \\ \vdots \\ -e^{2\varphi(T-1)} Q_{T-1}'Q_{\mathds{1}}
\end{bmatrix}.
\end{equation*}
Thus, the information matrix for $w$ is
\begin{align*}
\mathcal{I}_w 
&= R - K'\left(H_{d-1}\right)^{-1}K \\
&= \sum_{t=0}^{T-1}e^{2t\varphi} (I_{d-1} - Q_{\mathds{1}}' Q_{t}Q_{t}'Q_{\mathds{1}}).
\end{align*}

\begin{claim}
	Let $A \in \mathds{R}^{1\times n}$ be a row matrix and define $\hat A = \frac{A}{\|A\|}$ as the unit vector in the direction of $A$. Set $\hat Q \in \N(A)$. Then
	\begin{equation}
	\hat{Q}\hat{Q}' = I_n - \hat{A}'\hat{A}.
	\end{equation}
\end{claim}

\begin{proof}
	Let $x \in \mathds{R}^n$ be an arbitrary vector. By definition of the nullspace, the columns of $\hat Q$ and $\hat{A}'$ form an orthonormal basis for $\mathds{R}^n$. Hence, there exists vectors $x_Q$ and $x_A$ such that
	\begin{equation}
	x = \hat Qx_Q + \hat{A}'x_A
	\end{equation}
	This gives
	\begin{align*}
	\hat{Q}\hat{Q}' x
	&= \hat{Q}\hat{Q}'\hat Qx_Q + \hat{Q}\hat{Q}'\hat{A}'x_A \\
	&= \hat{Q}x_Q.
	\end{align*}
	On the other hand,
	\begin{align*}
	(I_n - \hat{A}'\hat{A}) x
	&= \hat Qx_Q + \hat{A}'x_A - \hat{A}'\hat{A}\hat Qx_Q - \hat{A}'\hat{A}\hat{A}'x_A \\
	&=  \hat{Q}x_Q.
	\end{align*}
	Since $\hat{Q}\hat{Q}' x = (I_n - \hat{A}'\hat{A}) x$ for all $x \in \mathds{R}^n$, it follows that $\hat{Q}\hat{Q}' = I_n - \hat{A}'\hat{A}$.
\end{proof}

Let $\hat{A}_t = \frac{A_t}{\|A_t\|}$ be the unit vector in the direction of $A_t$. From the claim it follows
\begin{align}
\mathcal{I}_w
&= \sum_{t=0}^{T-1}e^{2t\varphi} (I_{d-1} - Q_{\mathds{1}}'(I_d - \hat{A}_t'\hat{A}_t)'Q_{\mathds{1}}) \nonumber\\
&= \sum_{t=0}^{T-1}e^{2t\varphi} Q_{\mathds{1}}'\hat{A}_t'\hat{A}_tQ_{\mathds{1}}. \label{eq:fisher_info}
\end{align}
The inverse of this matrix is the covariance for the maximum likelihood estimate of $w$. Thus, the square root of the largest principal component of $\Sigma_w = \left(\mathcal{I}_w\right)^{-1}$ is the estimate error for the eavesdropper.

The value of the estimate error cannot be cleanly extracted from this information matrix formulation. While the exponentially growing term appears concerning, the auxiliary principal component contributions of the matrices $Q_{\mathds{1}}'\hat{A}_t'\hat{A}_tQ_{\mathds{1}}$ exponentially shrink. The rate of this reduction is approximately equal to the unmasked system's convergence. The performance of the mask thus results as expected from the tension between the decay rate of the mask and convergence rate of the system. Figure \ref{fig:est_error_dist} illustrates  the best possible estimate errors for 1,000,000 randomly generated, masked FJ systems of 100 agents with degree 10, and decay rate 1. These errors are finite, so the influence structure is identifiable, however the expected estimate error is much greater than 1, so the influence structure is undiscoverable. The extent to which it is remains undiscoverable over varying system parameters is discussed using further numerical results in the following subsection.

\begin{figure}[htb!]
	\centering
	\includegraphics[width=.6\columnwidth]{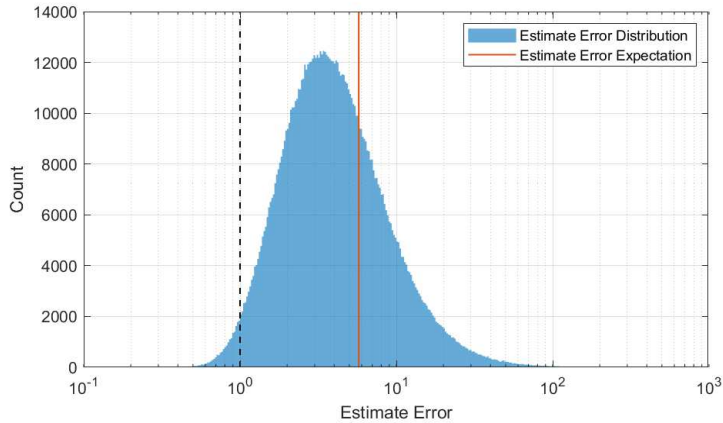}
	\caption{Distribution of the estimate error for random masked FJ systems}
	\label{fig:est_error_dist}
\end{figure}

\subsection{Discoverability Dependence on System Parameters}

Due to the complex and stochastic nature of $\mathcal{I}_w$, a deterministic lower bound on the estimate error could not be calculated. Therefore, we moved our analysis away from an analytical approach, and towards a numerical one.

To determine the dependencies of the estimate error, we simulated many masked FJ systems holding the hyper-parameters constant, except the one under analysis. We tuned this parameter over a range of practical values and recorded the results. When each of the global parameters were not being altered, we set the number of agents to 100, the degree of the system to 10, the decay rate to 0.3, the convergence tolerance to $10^{-4}$, and the initial opinions and susceptibilities were selected from the uniform distribution [0,1]. The convergence tolerance is the maximum allowable change of the agents' opinions over a single iteration for the system to be considered converged.

A comparison of the spreads of the estimate error against the decay rate of the mask is shown in Figure \ref{fig:decay_dependence}. Interestingly, there is an optimal decay rate for privacy preservation. For an FJ system defined by the previous global parameters, the box plots show that the decay rate which maximises the privacy of this system is $\varphi = 1$. The expected estimate error at this decay rate was 5.19, thus we concluded in this case the influence structure was undiscoverable.

\begin{figure}[ht!]
	\centering
	\includegraphics[width=.6\columnwidth]{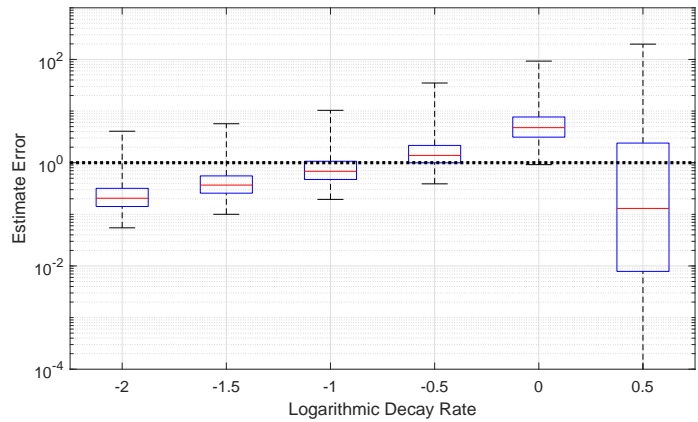}
	\caption{Spread of estimate errors against logarithmic decay rate of the mask ($\log_{10}\varphi$)}
	\label{fig:decay_dependence}
\end{figure}

The decline in privacy as the decay rate decreased was due to the slower convergence of the system as the mask perturbed the opinions longer. This delayed the opinions from converging to within the thresholds, so the eavesdropper had access to more timesteps to extract system information. On the other hand, as the decay rate increased beyond unity, the noise injected by the mask decayed too fast revealing the true trajectories before they had sufficiently converged.


A comparison of the spreads of the estimate error against the centre point of the distribution from which the agent susceptibilities were selected is shown in Figure \ref{fig:sus_dependence}. In this case, the susceptibilities were randomly selected from a uniform distribution of width 0.1 with midpoints ranging from 0.05 to 0.95.

\begin{figure}[ht!]
	\centering
	\includegraphics[width=.6\columnwidth]{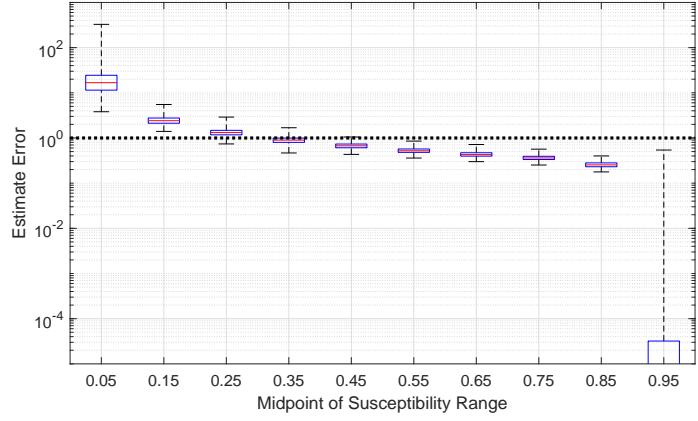}
	\caption{Spread of estimate errors against the susceptibility midpoint}
	\label{fig:sus_dependence}
\end{figure}

The estimate error of the system decreased as the susceptibilities of the agents increased. Additionally, the decrease in privacy occurred at a decreasing rate until the susceptibilities neared a value of 1, where the worst-case estimate error cascaded towards zero. As the susceptibilities decreased, the eigenvalues of $\Lambda W$ grew smaller, hence the unmasked FJ system converged faster. Therefore, the convergence of the FJ system was primarily dominated by the mask for lower susceptibilities, so it was difficult for the eavesdropper to isolate the opinion trends resulting from the true influence structure. As the susceptibilities grew, the convergence of the masked system became increasingly dependent on the underlying influence structure, making it more susceptible to discovery.



Figure \ref{fig:degree_dependence} illustrates that the estimate error increased exponentially with the degree of the system. The largest exponential rate of increase occurred when the degree was less than 15, where the median of the estimate errors rapidly approached 1. Therefore, a larger degree was more beneficial to prevent the discovery of the influence structure.

\begin{figure}[ht!]
	\centering
	\includegraphics[width=.6\columnwidth]{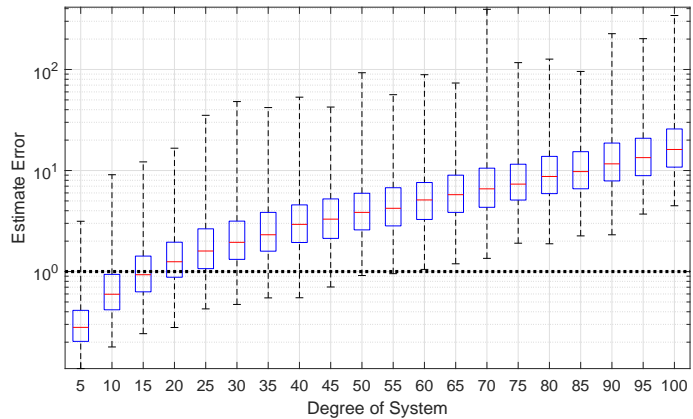}
	\caption{Spread of estimate errors against the degree of the system}
	\label{fig:degree_dependence}
\end{figure}

Figure \ref{fig:degree_dependence} also displays the importance of the availability of the underlying network to the eavesdropper when it performs the maximum likelihood estimate. Without knowledge of the graph, the problem simplification in Section \ref{sec:eavesdropper} where the non-influence elements of $W$ were discarded could not be performed. This is because the influence elements of $W$ are specified by the underlying network. Therefore, the eavesdropper would have no option but to extend the ``degree'' of the estimate of the influence structure from $d$ to $n$. Provided the estimate of the influence structure $\hat w$ was accurate, the entries of $\hat w$ close to zero could be assumed disconnections in the network. However, the number of agents in a distributed system is commonly orders of magnitude larger than the degree of the system, so the estimate error would rapidly increase, rendering the influence structure undiscoverable.

\section{Conclusions} \label{sec:conclusion}
In this paper, we showed that the convergence behaviour of an FJ system was preserved under a much broader set of conditions than what is currently appreciated in literature. Rather than requiring the agents' susceptibilities, biases and influence structure to be constant with respect to time, we demonstrated that provided these parameters converged to their true values, the limiting opinions of the agents were preserved. We used this result to develop a mask for the influence structure of an FJ system. The privacy preserving performance of the mask was dependent on global system parameters, so numerical simulations were required to determine its validity on a case by case basis. Nevertheless, a high degree, concealed underlying network, low convergence tolerance, or low susceptibility range gave strong indicators that the privacy of the influence structure was preserved by the mask. 

In future work, the formulation of the mask could be refined to improve the numerical results or enforce strict bounds on the estimate error. Additionally, other masks for the susceptibilities, external biases, or initial opinions of the agents could be developed using the liberties granted by Theorem \ref{thm:fjtheorem}. The compatibility of these masks could be analysed to determine if they work in unison to preserve the privacy of multiple parameters of an FJ system.



\end{document}